\documentclass[11pt,a4paper]{article}
\usepackage{times,latexsym}
\usepackage{url}
\usepackage[T1]{fontenc}
\usepackage[acceptedWithA]{tacl2018v2}

\usepackage{amsmath}
\usepackage{amssymb}
\usepackage{amsfonts}
\usepackage{amsthm}

\usepackage{booktabs}

\newtheorem{theorem}{Theorem}
\newtheorem{lemma}{Lemma}
\newtheorem{proposition}{Proposition}
\newtheorem{corollary}{Corollary}
\theoremstyle{definition}
\newtheorem{definition}{Definition}

\theoremstyle{remark}

\DeclareMathOperator{\bin}{bin}
\DeclareMathOperator{\att}{att}
\DeclareMathOperator{\act}{act}
\DeclareMathOperator{\pool}{pool}
\DeclareMathOperator{\uha}{UHA}
\DeclareMathOperator{\aha}{AHA}
\newcommand{\layers}{K}
\newcommand{\parity}{\texttt{PARITY}}
\newcommand{\majority}{\texttt{MAJORITY}}
\newcommand{\equality}{\texttt{EQUALITY}}
\newcommand{\dyck}{\texttt{DYCK}}
\newcommand{\shuffle}{\texttt{SHUFFLE}}
\newcommand{\palindromes}{\texttt{PALINDROMES}}

\title{Formal Language Recognition by Hard Attention Transformers: Perspectives from Circuit Complexity}

\author{Yiding Hao, Dana Angluin, and Robert Frank \\
    Yale University \\
    New Haven, CT, USA \\
    {\tt firstname.lastname@yale.edu} \\
}

\begin{document}
\maketitle
\begin{abstract}
    This paper analyzes three formal models of Transformer encoders that differ in the form of their self-attention mechanism: \textit{unique hard attention} (UHAT); \textit{generalized unique hard attention} (GUHAT), which generalizes UHAT; and \textit{averaging hard attention} (AHAT).  We show that UHAT and GUHAT Transformers, viewed as string acceptors, can only recognize formal languages in the complexity class AC$^0$, the class of languages recognizable by families of Boolean circuits of constant depth and polynomial size. This upper bound subsumes \citeauthor{Hahn20}'s (\citeyear{Hahn20}) results that GUHAT cannot recognize the \dyck\ languages or the \parity\ language, since those languages are outside AC$^0$ \citep{FSS84}. In contrast, the non-AC$^0$ languages \majority\ and \dyck-$1$ are recognizable by AHAT networks, implying that AHAT can recognize languages that UHAT and GUHAT cannot.
\end{abstract}

\section{Introduction}
\label{sec:introduction}

The Transformer architecture for neural networks \citep{Vaswanietal2017} has yielded remarkable advances in performance on a variety of benchmark tasks in natural language processing. These advances have spurred considerable interest in understanding the capabilities and limitations of the Transformer architecture. While Transformer networks are extremely complex when deployed at scale, theoretical studies such as those of \citet{PMB2019}, \citet{YRBRK2019}, \citet{Hahn20}, and \citet{MGSS2021} have uncovered meaningful insights about the expressive power of Transformers by formulating abstract models of the self-attention mechanism and analyzing their computational power. 

In this work, we analyze three restricted models of self-attention based on their ability to recognize formal languages. All three models use \textit{hard attention}---meaning that each attention head attends only to the position or positions with the highest attention score, with no attention paid to any of the other positions---but differ in how they behave in the case of ties in the maximum attention value.  In the first two models we study,  the attention mechanism returns the value at exactly one position (for example, the leftmost) in case several positions tie for the maximum attention value. The first such model, which we call \emph{generalized unique hard attention Transformers} (GUHAT) and was defined by \citet{Hahn20}, imposes no restrictions on the nature of activation values or the functions the network uses to compute them.  The second model, \emph{unique hard attention Transformers} (UHAT), was defined and studied by \citet{YPPN2021} and is a more concrete version of GUHAT that incorporates restrictions on the nature of activation values and  computations. In the third model, which we call \emph{averaging hard attention Transformers} (AHAT), the attention mechanism returns the uniform average of the values at positions with the maximum attention value.  This is the definition of hard attention used by \citet{PMB2019}, \citet{YRBRK2019}, and \citet{MGSS2021}.\footnote{\citet{MGSS2021} call it \emph{saturated hard attention}.}

Our main contribution is to prove that GUHAT and UHAT can only recognize formal languages in AC$^0$, the class of formal languages recognized by a family of Boolean circuits of constant depth and polynomial size, whereas AHAT can recognize formal languages outside of AC$^0$.  More formally, we prove that any formal language recognized using a GUHAT is also recognized by a family of Boolean circuits of constant depth and polynomial size, establishing AC$^0$ as an upper bound on the expressiveness of UHAT and GUHAT. We also show that every UHAT can be simulated by an AHAT, establishing UHAT as a subclass of AHAT. Based on the classical results of \citet{FSS84}, our upper bound subsumes \citeauthor{Hahn20}'s (\citeyear{Hahn20}) results that GUHAT cannot recognize the \dyck\ languages or the \parity\ language, neither of which belongs to AC$^0$. Furthermore, our result combines with \citeauthor{PMB2019}'s (\citeyear{PMB2019}) AHAT implementation of the \majority\ language and \citeauthor{bhattamishra-etal-2020-ability}'s (\citeyear{bhattamishra-etal-2020-ability}) AHAT implementation of \dyck-$1$
(neither of which is in AC$^0$) to show that AHAT can recognize languages that GUHAT cannot. Recently, \citet{MGSS2021} have given an upper bound on the power of AHAT: namely, that every formal language recognizable using averaging hard attention is recognizable using a family of circuits of constant depth and polynomial size with Boolean and majority gates; that is, a family of circuits in the complexity class TC$^0$, known to be a strict superset of AC$^0$. Taken together, our paper establishes the following relationships between the three models we consider of hard-attention Transformers.
\begin{align*}
    \text{UHAT} \subseteq \text{GUHAT} &\subseteq  \text{AC}^0 \\
    \text{UHAT} \subsetneq  \text{AHAT} &\nsubseteq  \text{AC}^0 
\end{align*}

\section{Preliminaries}
\label{sec:preliminaries}

Let $\Sigma$ be a fixed finite alphabet of symbols, and
let $\$$ be a distinct end-of-sequence symbol not in $\Sigma$.
The set of strings of symbols over $\Sigma$ of length $n$ is denoted by $\Sigma^n$, and the 
set of all finite strings over $\Sigma$ is denoted $\Sigma^*$.
A (formal) language over $\Sigma$ is any subset of $\Sigma^*$.

The set of integers between $i$ and $j$ inclusive is denoted
$[i..j]$. Define the function $\ell:\mathbb{N} \to \mathbb{N}$ as $\ell(n) = \lceil \log_2(n+1) \rceil$ for all $n$, and define $\bin(i,n)$ to be the binary representation of $i$ as a string of length $\ell(n)$, for every $i \in [1..n]$. For example, $\bin(6, 30) = 00110$. If $P$ is a logical predicate, then $\{P\}$ denotes the truth value of $P$; for example, $\{(x \ge y) \vee (z > 3)\}=1$ when  $x \ge y$ or $z>3$, and is $0$ otherwise.

\section{Circuit Complexity}

Our analysis of UHAT and GUHAT is carried out within the framework of \textit{circuit complexity}, in which the complexity of a computational system is measured by the size, depth, and types of gates of a Boolean circuit implementing that system. In this section we review the basic concepts, definitions, and results of circuit complexity used by our analysis. A detailed overview is provided in Chapter 6 of \citet{aroraComputationalComplexityModern2009}.

\subsection{Boolean Circuits}

Boolean circuits are a formal model of computational systems based on logic gates. Roughly speaking, a Boolean circuit consists of binary-valued input and output layers, with feedforward connections\footnote{We consider only acyclic circuits.} to one another via intermediate \textit{gates} that implement logical operations. We use the following definition of Boolean circuits.

\begin{definition}
    A \textit{Boolean circuit with $n$ inputs and $m$ outputs} is a labeled directed acyclic graph satisfying the following conditions.
    There are $n$ distinguished \textit{input vertices} labeled with the variables $x_1, x_2, \ldots, x_n$. Each input vertex has fan-in $0$.  The rest of the vertices are \textit{gates}, each having a label from Constant-0, Constant-1, NOT, AND, or OR.  The Constant-0 and Constant-1 gates have fan-in $0$, NOT gates have fan-in $1$, and AND and OR gates have unbounded fan-in.  Finally, the labels $z_1, z_2, \ldots, z_m$ are applied to some (not necessarily distinct) vertices; these are the \textit{outputs}.
    
    We refer to the edges of a Boolean circuit as \textit{wires}. The \emph{size} of a circuit is the number of wires it contains, and the \emph{depth} of a circuit is the maximum length of a directed path of wires from an input vertex to an output. A Boolean circuit $C$ computes a Boolean function from $\lbrace 0, 1 \rbrace^n$ to $\lbrace 0, 1 \rbrace^m$; we denote its output on input $x = x_1x_2 \dots x_n$ by $C(x)$.
\end{definition}

Observe that a Boolean circuit has a fixed number of input vertices, and therefore can only take as input bit strings of a fixed length. We would like to define circuit computation for a map defined on all of $\lbrace 0, 1 \rbrace^*$. To that end, we allow different circuits for inputs of different lengths.
\begin{definition}
    A \emph{family of circuits} is a sequence $\{C_n\}$, where for each integer $n \ge 0$, $C_n$ is a Boolean circuit with $n$ inputs and one output. A map $f$ from $\{0,1\}^*$ to $\{0,1\}$ is computed by a family of circuits $\{C_n\}$ if and only if for all $n$ and all $x \in \{0,1\}^n$, $f(x) = C_n(x)$.
\end{definition}
The class AC$^0$ is defined by setting restrictions on the size and depth of circuits within a family of Boolean circuits.
\begin{definition}
    A family of circuits is \emph{of constant depth} if there exists a constant $K$ such that the depth of $C_n$ is bounded by $K$ for all $n$. A family of circuits is \emph{of polynomial size} if there exists a constant $c$ such that the size of $C_n$ is bounded by $n^c + c$ for all $n$. The set AC$^0$ is the set of families of Boolean circuits of both constant depth and polynomial size. 
\end{definition}
We relate formal languages with families of circuits by identifying languages $L \subseteq \Sigma^*$ with Boolean functions that classify strings as belonging to $L$ or not. Formally speaking, let $\Sigma$ be a finite alphabet. A \emph{binary symbol encoding} of $\Sigma$ is an injective map $h$ from $\Sigma$ to $\{0,1\}^s$, where $s = \ell(|\Sigma|)$.  Thus $h$ maps each symbol to a distinct binary string of length $s$. We extend $h$ to a homomorphism on strings from $\Sigma^*$. We say that the circuit family $\{C_n\}$ \textit{recognizes the language $L$ over $\Sigma$} if there is a binary symbol encoding $h$ of $\Sigma$ such that for every $n$ and every $x \in \Sigma^n$, $C_{sn}(h(x)) = 1$ if and only if $x \in L$. With this definition of language recognition via Boolean circuits, we say that a language is \textit{in AC$^0$} if and only if it is recognized by a family of Boolean circuits in AC$^0$.

\subsection{Non-AC$^0$ Languages}

Having defined the class AC$^0$, we present some examples of languages \textit{not} belonging to this class. First, the following three languages were shown by \citet{FSS84} to fall outside AC$^0$.
\begin{definition}
    We define the following languages over the alphabet $\lbrace 0, 1 \rbrace$. The language \parity\ is the set of all strings containing an even number of $1$s; \majority\ is the set of strings with at least as many $1$s as $0$s; and \equality\ is the set of strings with exactly as many $1$s as $0$s.
\end{definition}

Additionally, we show later in this paper (\autoref{Dyck-languages-not-in-AC0}) that \dyck-$1$ also falls outside AC$^0$.
\begin{definition}
    The language \dyck-$k$ is the set of strings over an alphabet of $k$ types of pairs of brackets that are correctly nested and matched. For example, \dyck-$2$ over the alphabet $\{(,),[,]\}$ can be described by a context free grammar with productions $S \to \varepsilon$, $S \to (S)$, $S \to [S]$, and $S \to SS$. The language \dyck-$(k,D)$ is the set of strings in \dyck-$k$ in which the depth of nesting of brackets never exceeds $D$. The language \shuffle-$k$ is the shuffle (arbitrary interleaving) of strings from $k$ versions of \dyck-$1$ each using a different type of bracket pair.
 \end{definition}
 
 Finally, we define \palindromes, a language shown in \autoref{sec:palindrome-example} to be in GUHAT.

 \begin{definition}
    The language \palindromes\ is the set of strings equal to their reverses, which can be described by the context free grammar with productions $S \to \varepsilon$, $S \to \sigma$, and $S \to \sigma S \sigma$ for each alphabet symbol $\sigma$.
\end{definition}

\section{Hard Attention Transformers}
\label{sec:definitions-hard-attention-transformers}

We now define the three kinds of hard attention Transformers studied in this paper: GUHAT, UHAT, and AHAT. These formalisms are models of computation inspired by the encoder portion of the Transformer architecture. They conceptualize Transformers as cascading layers of feature extractors that convert a sequence of embeddings into increasingly higher-level representations. 


\subsection{General Framework}

We begin by presenting a general framework that subsumes the three hard attention Transformer models. Formally, a generalized Transformer is a device that maps a string $x \in \Sigma^*\$$ to $1$ or $0$, signifying that $x$ is \textit{accepted} or \textit{rejected}, respectively. Each such device is parameterized by a collection of functions described as follows.
\begin{definition}
    A \textit{generalized Transformer with $K$ layers and $H$ attention heads} is a tuple $T = (\Sigma, \mathcal{A}, f, f^{\att}_{k, h}, f^{\pool}, f^{\act}_k, g\,  \mid  k \in [1..K], h \in [1..H])$ where
    \begin{itemize}
        \item $\Sigma$ is the \textit{input alphabet},
        
        \item $\mathcal{A}$ is the set of \textit{activation values},
        
        \item $f:\Sigma \cup \lbrace \$ \rbrace \times \mathbb{N} \times \mathbb{N} \to \mathcal{A}$ is the \textit{input function}, 
        
        \item $f_{k,h}^{\att}:\mathcal{A} \times \mathcal{A} \to \mathbb{R}$ is the \textit{attention function for layer $k$ and head $h$},
        
        \item $f^{\pool}:\mathcal{A}^* \times \mathbb{R}^* \to \mathcal{A}$ is the \textit{pooling function},
        
        \item $f_k^{\act}:\mathcal{A}^{H + 1} \to \mathcal{A}$ is the \textit{activation function for layer $k$}, and
        
        \item $g:\mathcal{A} \to \lbrace 0, 1 \rbrace$ is the \textit{model output function}.
    \end{itemize}
\end{definition}

On input $x_1x_2 \dots x_n$ where $x_n = \$$, a string $y^{(0)}_1y^{(0)}_2 \dots y^{(0)}_n \in \mathcal{A}^n$ of \textit{initial activation values} is given by
\[
y_i^{(0)} = f(x_i, i, n)
\]
for all $i$. Each layer $k$ then produces a string $y^{(k)} = y^{(k)}_1\allowbreak y^{(k)}_2\allowbreak \dots \allowbreak y^{(k)}_n\allowbreak \in \mathcal{A}^n$ of activation values from the previous activation values $y^{(k - 1)} = y^{(k - 1)}_1y^{(k - 1)}_2 \dots y^{(k - 1)}_n$ as follows. First, each attention head $h$ produces an $n \times n$ matrix of \textit{attention scores} $a_{i, j, k, h}$ given by
\[
    a_{i,j,k,h} = f_{k,h}^{\att}\left(y_i^{(k-1)},y_j^{(k-1)}\right)
\] 
for all positions $i, j$ of the input string. Next, the pooling function converts each row of attention scores into an activation value based on $y^{(k - 1)}$:
\begin{align*}
    b_{i, k, h} &= f^{\pool}\Big(y^{(k - 1)}_1, y^{(k - 1)}_2, \dots, y^{(k - 1)}_n, \\
    &\phantom{= f^{\pool}\Big(} a_{i,1,k,h}, a_{i,2,k,h},  \dots, a_{i,n,k,h}\Big)\text{.}
\end{align*}
Finally, the layer output $y^{(k)}$ is computed using the layer's activation function:
\[
    y_i^{(k)} = f_k^{\act}\left(y_i^{(k-1)},b_{i,k,1},\ldots,b_{i,k,H}\right)\text{.}
\]
When $y^{(k)}$ has been computed for all $k \in [1..K]$, the final output of the generalized Transformer $T(x)$ is computed by applying the model output function to the last symbol of $y^{(K)}$; that is,
\[
T(x) = g\left(y_n^{(K)}\right)\text{.}
\]
If $T(x) = 1$, we say that \textit{$T$ accepts $x$}; otherwise, we say that \textit{$T$ rejects $x$}. The \textit{language recognized by $T$}, denoted $L(T)$, is the set of strings $x \in \Sigma^*$ such that $T$ accepts $x\$$.

The formalism we have presented above is fully generalized in the sense that we have placed no restrictions on the activation values $\cal{A}$ or the functions $f$, $f_{k,h}^{\att}$, $f^{\pool}$, $f_k^{\act}$, or $g$, other than to specify their domains and co-domains. The three hard attention Transformer models are derived by placing restrictions upon these elements.

\subsection{Unique and Averaging Hard Attention}

The first restriction we consider is on the form of the pooling function. We consider two types of pooling functions: the \textit{unique hard attention} function, used in GUHAT and UHAT, and the \textit{averaging hard attention} function, used in AHAT.

In unique hard attention, the pooling function simply selects the activation value from the previous layer corresponding to the argmax of the row of attention scores. In case of a tie, the leftmost activation value is selected.
\begin{definition}
    The \textit{unique hard attention function} is the pooling function $f^{\uha}:\mathcal{A}^* \times \mathbb{R}^* \to \mathcal{A}$ defined as follows. On inputs $(y_1, y_2, \dots, y_n) \in \mathcal{A}^n$ and $(a_1, a_2, \dots, a_n) \in \mathbb{R}^n$, let $j \in [1..n]$ be the least position that maximizes $a_j$. Then,
    \[
    f^{\uha}(y_1, \dots, y_n, a_1, \dots, a_n) = y_j.
    \]
\end{definition}

Averaging hard attention is similar to unique hard attention, except that in the case of a tie, the selected activation values are averaged.
\begin{definition}
    Let $\mathcal{A}$ be a vector space over a field containing $\mathbb{Q}$.  The \textit{averaging hard attention function} is the pooling function $f^{\aha}:\mathcal{A}^* \times \mathbb{R}^* \to \mathcal{A}$ defined as follows. On inputs $(y_1, y_2, \dots, y_n) \in \mathcal{A}^n$ and $(a_1, a_2, \dots, a_n) \in \mathbb{R}^n$, let $j_1, j_2, \dots, j_m \in [1..n]$ be all the positions that maximize $a_{j}$. Then,
    \[
    f^{\aha}(y_1, \dots, y_n, a_1, \dots, a_n) = \frac{1}{m} \sum_{i = 1}^m y_{j_i}\text{.}
    \]
\end{definition}

The GUHAT model is defined as the class of generalized Transformers that use unique hard attention.
\begin{definition}
    A \textit{generalized unique hard attention Transformer} is a generalized Transformer whose pooling function is $f^{\uha}$. We use the term \textit{GUHAT} to refer to the class of generalized unique hard attention Transformers, and also to the class of languages they recognize. 
\end{definition}
The GUHAT model mostly follows the definitions of \citet{Hahn20}. It is slightly generalized in allowing the input function $f$ to depend on the input length $n$, and in allowing the activation function $f_k^{\act}$ to depend on the layer $k$, but these generalizations are immaterial. In particular, it is not necessary to assume that the input length $n$ is provided to the input function: if the input function were $f(\sigma,i) = (\sigma,i)$, the subsequent layer could direct attention at every position to position $n$ (because it uniquely contains the end-of-sequence symbol $\$$), at which point the value of $n$ is available at every position.

\subsection{Restricted Models: UHAT and AHAT}

GUHAT allows the activation values $\cal{A}$ and the functions $f$, $f_{k,h}^{\att}$, $f_k^{\act}$, and $g$ to take on \textit{any} arbitrary mathematical value. In practical applications of Transformer networks, however, these components are restricted in specific ways. Many variations of hard attention Transformers attempt to incorporate these restrictions into theoretical models, though they do not entirely agree on the details of these restrictions.

The UHAT and AHAT models adopt many of these restrictions, largely following the definitions of \citet{YPPN2021}. For the sake of computability, we require activation values to be vectors of rational numbers. Following \citet{PMB2019}, we restrict scalars to be rational as well. Next, we assume that the input function $f$ is decomposed into a token embedding function and a position embedding function. Mirroring the more familiar description of attention functions in terms of query, key, and value matrices, we use a bilinear form for attention functions proposed in \citet{LPM2015}. In addition to the unique and averaging hard attention mechanisms, we allow the pooling function to be \textit{future-masked} (where for position $i$ only those positions $j$ with $j \le i$ are considered in the attention computation) or \textit{past-masked} (similarly for $j \ge i$). Finally, we assume that activation functions and the model output function are computed by feedforward neural networks with ReLU activation. These restrictions are summarized below.
\begin{definition}
    For $d \in \mathbb{N}$, a \textit{restricted Transformer of dimension $d$} is a generalized Transformer such that
    \begin{itemize}
        \item the set of activation values is $\mathcal{A} = \mathbb{Q}^d$;
        
        \item the input function is given by
        \[
        f(\sigma, i, n) = f_e(\sigma) + p(i, n)\text{,}
        \]
        where $f_e:\Sigma \cup \lbrace \$ \rbrace \to \mathbb{Q}^d$ is the \textit{token embedding function} and $p:\mathbb{N} \times \mathbb{N} \to \mathbb{Q}^d$ is the \textit{position embedding function};
        
        \item each attention function is of the form
        \[
        f_{k,h}^{\att}(y, y^\prime) = y^\top A_{k,h} y^\prime\text{,}
        \]
        where $A_{k, h} \in \mathbb{Q}^{d \times d}$;
        
        \item the pooling function may be \textit{future-masked} or \textit{past-masked};
        
        \item each activation function is computed by a feedforward neural network with ReLU activation;
        
        \item the output function $g$ is computed by a feedforward neural network with ReLU activation followed by a softmax layer, with $g(y) = 1$ if and only if the output of the network on input $y$ is greater than or equal to $1/2$.
    \end{itemize}
\end{definition}

Because $\Sigma$ is finite, we may assume that the token embedding function is given by a table lookup. Our formulation of position embedding is somewhat more general than the definition of \citet{YPPN2021}, who take the position embedding to be a scalar defined as $p(i,n) = i/n$ that occupies one position of the initial activation vector.

The UHAT and AHAT models are defined to be restricted Transformers that satisfy the above conditions and use unique and averaging hard attention, respectively.
\begin{definition}
    A \textit{unique hard attention Transformer} is a restricted Transformer whose pooling function is $f^{\uha}$ or a future- or past-masked version thereof. An \textit{averaging hard attention Transformer} is a restricted Transformer whose pooling function is $f^{\aha}$ or a future-or past-masked version thereof. We use the terms \textit{UHAT} and \textit{AHAT}, respectively, for these classes of Transformers, and also for the classes of languages they recognize.
\end{definition}

UHAT is clearly a subclass of GUHAT because the former imposes restrictions on the form of the input, attention, activation, and output functions. We suspect, but do not prove, that this inclusion is proper. Moreover, we briefly argue below that UHAT is properly contained in AHAT. 
\begin{proposition}
\label{lemma:uhat-in-ahat}
UHAT is a strict subclass of AHAT.
\end{proposition}
\begin{proof}[Proof sketch]
Since AHAT recognizes non-AC$^0$ languages \citep{PMB2019,bhattamishra-etal-2020-ability}, it suffices to show that $\text{UHAT} \subseteq \text{AHAT}$. Let $T$ be a UHAT of dimension $d$ recognizing $L$. We define a UHAT $\hat{T}$ of dimension $d+2$ recognizing $L$ that has no ties in its attention values. Since the pooling functions $f^{\uha}$ used in UHAT and $f^{\aha}$ used in AHAT are identical in the absence of ties, replacing the pooling function of $\hat{T}$ with $f^{\aha}$ gives us an AHAT recognizing $L$.

Let $N$ be a sufficiently large integer depending on $n$, specified below.
Each activation value $\hat{y}_i^{(k)}$ in $\hat{T}$ is $y_i^{(k)}$ from $T$ with two additional constant components, set to $1$ and $i/N$ by the input function.  The attention function $\hat{f}_{k,h}^{\att}(\hat{y}_i^{(k-1)}, \hat{y}_j^{(k-1)})$ computes $a_{i,j,k,h}$ using the original attention function and activation values, subtracting the value $j/N$. This is achievable with a bilinear map.

If for $j < \ell$ the attention values $a_{i,j,k,h}$ and $a_{i,\ell,k,h}$ are tied in $T$, then after subtracting $j/N$ and $\ell/N$ respectively, the tie is broken in favor of $j$.
However, $N$ must also be large enough to preserve the order of any two attention values that are not tied.  There are a finite number of different attention values that arise in the computation of $T$ on all the inputs of length $n$, and it suffices to choose $N$ so that $n/N$ is less than the distance between any pair of such attention values.
\end{proof}

\subsection{Prior Results for These Models}

\citet{Hahn20} shows that the languages $1^*$ and $\{a^nb^n \mid n \ge 1\}$ are in GUHAT, and the languages \parity\ and \dyck-$k$ for all $k \ge 1$ are not in GUHAT.
\citet{PMB2019} show that even without positional information, the language \majority\ is in AHAT.
\citet{bhattamishra-etal-2020-ability} show that \shuffle-$k$ is in AHAT, which implies that \dyck-$1$ is in AHAT.
\citet{YPPN2021} show that the language \dyck-$(k,D)$ is in UHAT.
The latter two results use positional masking, but no other positional information.

\section{\palindromes\ in GUHAT}
\label{sec:palindrome-example}

\begin{figure*}
    \centering
    \footnotesize
    \begin{tabular}{l | r c c c c c c}
            \toprule
             \textbf{Input} & $x \mathrel{=}$ & $a$ & $b$ & $c$ & $c$ & $a$ & $\$$ \\\midrule
             Initial Activation Values & $y^{(0)} \mathrel{=}$ & $(a, 1, 6)$ & $(b, 2, 6)$ & $(c, 3, 6)$ & $(c, 4, 6)$ & $(a, 5, 6)$ & $(\$, 6, 6)$ \\
             Layer 1 Activation Values & $y^{(1)} \mathrel{=}$ & $(0, 1)$ & $(1, 2)$ & $(0, 3)$ & $(1, 4)$ & $(0, 5)$ & $(1, 6)$ \\
             Layer 2 Activation Values & $y^{(2)} \mathrel{=}$ & $(1, 2)$ & $(2, 2)$ & $(3, 2)$ & $(4, 2)$ & $(5, 2)$ & $(6, 2)$ \\\midrule
             \textbf{Output} & $g\left(y^{(2)}_6\right) \mathrel{=} $ & & & & & & $0$ \\\bottomrule
    \end{tabular}
    \caption{Activation values computed by a GUHAT Transformer for \palindromes\ as it rejects the input $abcca\$$.}
    \label{fig:abcca}
\end{figure*}
Let us now illustrate how a GUHAT computes by way of example. In this section, we describe a GUHAT with $2$ layers and $1$ head that recognizes the language \palindromes\ over the alphabet $\Sigma = \{ a, b, c \}$. Broadly speaking, this Transformer works as follows. The first layer is responsible for comparing each symbol of the input string with the corresponding symbol on the opposite side of the string, and marking whether the two symbols match. The second layer reads these markings, searching for a mismatch identified by the first layer. If one is found, the model output function returns $0$; otherwise, it returns $1$. For intuition, we simultaneously illustrate the Transformer's computation on the input $abcca\$$, which should be rejected.

The input function is defined as
$f(\sigma,i,n) = (\sigma, i, n)$ 
for each $\sigma \in \Sigma$ and $i \in [1..n]$.
For our example input, the initial (layer $0$) activation values are shown in the first row of \autoref{fig:abcca}.
These activation values are not rational-valued vectors, of course, but the GUHAT model imposes no restriction on the form these values can take.

We define the attention function for layer $1$,
$f_{1,1}^{\att}((\sigma,i,n), (\sigma',j,n))$, to be $\{(j = n-i) \vee (i = j = n)\}$.
For each position $i < n$, this selects the activation at the correct corresponding position, $n-i$.  For position $n$, it selects the activation at position $n$.

We define the activation function for layer $1$ as 
\[f^{\act}_1((x_i, i, n), (x_j, j, n)) = 
\begin{cases}
    (1, i), & x_i \neq x_j \\ 
    (1, i), & i = j = n \\
    (0, i), & \text{otherwise.}
\end{cases}
\]

The layer $1$ activation values for our example input are shown in the second row of \autoref{fig:abcca}.
This indicates that positions $2$ and $4$ found mismatched symbols, and positions $1$, $3$, and $5$ did not.

Layer $2$ gathers the relevant information from layer $1$ into the last position. The layer $2$ attention function is defined by $f_{2,1}^{\att}((r, i), (s, j)) = s$. This directs the attention at every position to the leftmost activation value $(s,j)$ from layer $1$ such that $s = 1$.
In our example, the leftmost such position is $2$, with its activation of $(1,2)$.
If the input sequence had instead been a valid palindrome, none of the positions $i \in [1..n - 1]$ would have been marked with $(1, i)$ by layer $1$. 
In this case, the leftmost position with $s = 1$ would have been the final position $n$, which has the activation value of $(1, n)$.

We define the layer $2$ activation function as
$f_2^{\act}((r, i), (s, j)) = (i, j)$.
For our example input, the activation values for layer $2$ are shown in the third row of \autoref{fig:abcca}.
The activation value at position $n$ will be $(n,n)$
if and only if no earlier position found a symbol mismatch, so the model output function is simply
$g((i,j)) = \{i = j \}$.
With the input sequence $abcca\$$, the activation for position $6$ at layer $2$ is $(6,2)$ and the output value is $0$. For a valid palindrome such as $abcba\$$, the activation for position $6$ at layer $2$ is $(6,6)$ and the output value is $1$. Generalizing this construction to an arbitrary alphabet $\Sigma$, we have the following.
\begin{proposition}
\label{lemma:palindromes}
For any finite alphabet $\Sigma$, the language  \palindromes\ over $\Sigma$ is in GUHAT. 
\end{proposition}

\section{A Normal Form for GUHAT}
\label{sec:normal-form}

Despite the abstractness and generality of the GUHAT model, we can define a normal form representation and show that every Transformer $T$ in GUHAT is equivalent to a  Transformer in GUHAT in this normal form with the same number of layers and heads.
The key idea is to preserve in the activation values all the information from previous layers that has been used to compute them,
by requiring that the input and activation functions just return the tuple of their arguments.
We also require that attention values be integers in the smallest relevant range.

\begin{definition}
A GUHAT with $K$ layers and $H$ heads is in \emph{informative normal form} if and only if the following conditions are satisfied.
\begin{itemize}
    \item The input function is
    $f(\sigma,i,n) = (\sigma,i,n)$.
    \item For each layer $k \in [1..\layers]$, the activation values are $(H+1)$-tuples of activation values at layer $k-1$, and the activation function is defined by
    \[
    f^{\act}_k(y,b_1, \ldots,b_H) = (y, b_1, \ldots, b_H)\text{.}
    \]
    \item For each layer $k \in [1..\layers]$ and attention head $h \in [1..H]$, the attention function $f^{\att}_{k,h}$ returns an integer in $[0..N-1]$, where $N$ is the total number of possible ordered pairs of activation values at layer $k-1$.
\end{itemize}
\end{definition}

\begin{lemma}
\label{lemma:normal-form}
For any Transformer $T \in \text{GUHAT}$, there exists a Transformer $\hat{T} \in \text{GUHAT}$ in informative normal form such that $L(T) = L(\hat{T})$.  Moreover, $\hat{T}$ has the same number of layers and heads as $T$.
\end{lemma}

\begin{proof}
Let $T$ be a GUHAT with $\layers$ layers and $H$ heads, with input alphabet $\Sigma$, input function $f$, attention functions $f^{att}_{k,h}$, activation functions $f^{act}_k$, and output function $g$.  We describe how to construct functions for an equivalent Transformer $\hat{T}$ in GUHAT in informative normal form, which also has $\layers$ layers and $H$ heads.  We assume that $n$ is the input length.

For $\hat{T}$ the input function $\hat{f}(\sigma,i,n)$ is defined to return the triple $(\sigma,i,n)$.
Note that there are at most $|\Sigma|n$ possible initial activation values.
We also define a function $t_0$ that translates initial activation values for $\hat{T}$ into initial activation values for $T$ by
$t_0(\sigma,i,n) = f(\sigma,i,n)$.

Now, we induct on the layers of $T$ and $\hat{T}$. Assume that we have defined attention and activation functions for $\hat{T}$ for layers before $k$ (where the initial activation values are treated as ``layer $0$''), and a translation function $t_{k-1}$ that translates all possible activation values for $\hat{T}$ from the previous layer into activation values for $T$ from the previous layer. To define the attention function for $\hat{T}$ for layer $k$ for head $h$, we enumerate all the possible pairs $\hat{y}_i$ and $\hat{y}_j$ of activation values of $\hat{T}$ at layer $k - 1$, and determine the corresponding attention values of $T$, which we denote by
$v_{k,h}(\hat{y}_i,\hat{y}_j) = f^{\att}_{k,h}(t_{k-1}(\hat{y}_i),t_{k-1}(\hat{y}_j))$.
We make a list of all the \emph{distinct} resulting values and sort them into increasing order.  Then we define $\hat{f}^{att}_{k,h}(\hat{y}_i,\hat{y}_j)$ to be the rank of $v_{k,h}(\hat{y}_i,\hat{y}_j)$ in this sorted list. The activation function for $\hat{T}$ for layer $k$ is, by definition,
\[
\hat{f}^{\act}_k(y, b_1,\ldots,b_H) = (y, b_1,\ldots,b_H)\text{.}
\]
The translation function for layer $k$ is defined by
\begin{align*}
&\mathrel{\phantom{=}} t_k(y, b_1,\ldots,b_H) \\
&= f_k^{act}(t_{k-1}(y),t_{k-1}(b_1),\ldots,t_{k-1}(b_H))\text{,}
\end{align*}
that is, we translate each of the component activation values using $t_{k-1}$ and then apply the activation function of $T$.

Finally, the output function for $\hat{T}$ is defined by $\hat{g}(\hat{y}) = g(t_\layers(\hat{y}))$, that is, we translate the layer $\layers$ activation value $\hat{y}$ of $\hat{T}$ to the layer $K$ activation value of $T$, and apply the output function of $T$.

By construction, $\hat{T}$ is in informative normal form, and it has $\layers$ layers and $H$ heads. It is not difficult to see that for any input $x$, the translations $t_k(\hat{y})$ of the activation values $\hat{y}$ of $\hat{T}$ are equal to the corresponding activation values of $T$, and the outputs $\hat{T}(x) = T(x)$ are equal as well. Thus, $L(\hat{T}) = L(T)$.
\end{proof}

To illustrate the construction of $\hat{T}$ in the proof of \autoref{lemma:normal-form}, we briefly show how an informative normal form version of the Transformer for \palindromes\ from \autoref{sec:palindrome-example} would process the input $x = abcca\$$. Because the attention functions in that example return $0$ or $1$, their translation is simplified.

The initial activation values and layer $1$ attention function are the same as in the example. The resulting layer $1$ activation sequence, consisting of a sequence of paired initial activations and attention values, is
\begin{align*}
((a,&1,6),(a,5,6)),\\ &((b,2,6),(c,4,6)), \ldots, ((\$,6,6),(\$,6,6)).
\end{align*}
The translation $t_1$ maps $((x_i,i,n), (x_j,j,n))$ to $(0,i)$ if $x_i = x_j$ and $(1,i)$ if $x_i \neq x_j$. When applied to the above activation sequence, this yields the previous example's layer $1$ activation sequence.

The layer $2$ attention function applied to a pair of layer $1$ activation values  $((x_i,i,n),(x_j,j,n))$ and $((x_k,k,n),(x_{\ell},\ell,n))$ first applies the translation function $t_1$ to these two activation values to recover the pairs $(r,i)$ and $(s,j)$, and then applies the example's layer $2$ attention function to these to yield the attention value $s$.

The layer $2$ translation function maps a layer $2$ activation value \[(((x_i,i,n),(x_j,j,n)),((x_k,k,n),(x_{\ell},\ell,n)))\] to $(i,k)$.
For layer $2$ and position $6$ the activation value for this input is
\[(((\$,6,6),(\$,6,6)),((b,2,6),(c,4,6))),\] which is mapped to $(6,2)$ by $t_2$.
The previous example's output function compares $6$ and $2$ and returns $0$, rejecting the input $x$.

\section{From GUHAT to Circuits}
\label{sec:guHAT-to-family-of-circuits}

In this section we show that for every language $L \in \text{GUHAT}$, we can construct a family of Boolean circuits of constant depth and polynomial size that also recognizes $L$.
This will prove the following, which is our main result.
\begin{theorem}
\label{theorem:guHAT0-in-AC0}
Every language in GUHAT is recognized by a family of circuits in AC$^0$.
\end{theorem}

Let $L$ be a language over $\Sigma$ that is in GUHAT.
By \autoref{lemma:normal-form}, we may assume that $L$ is recognized by GUHAT Transformer $T$ in informative normal form. Assume $T$ has $\layers$ layers and $H$ heads.

What we describe below is a family of circuits to recognize the end-marked language $L\$$, which can easily be converted to a family of circuits that recognizes $L$ by hard-wiring the representation of the end-of-sequence symbol $\$$ at the end of the input string using constant gates.
Let $s = \ell(|\Sigma|+1)$ and let $h$ be any binary symbol encoding for $\Sigma \cup \{\$\}$. We construct a family of Boolean circuits $\{C_n\}$ of constant depth and polynomial size such that for all positive integers $n$ and all $x \in \Sigma^{n-1}$, $x \in L$ if and only if $C_{sn}(h(x\$)) = 1$.

The key step of the proof is to bound the number of bits needed to represent attention and activation values for an input sequence of length $n$ by $O(\log n)$, where the suppressed constants depend on $\layers$ and $H$.

\begin{lemma}
\label{lemma:activation-attention-bit-bounds}
Let $T$ be a GUHAT in informative normal form with $\layers$ layers and $H$ heads, and alphabet $\Sigma$. Let $s = \ell(|\Sigma|+1)$.
Then for any input of length $n$ and any $k \in [0..\layers]$, the activation values at layer $k$ can be represented by $(H+1)^k (2\ell(n) + s)$ bits, and for $k \ge 1$, the attention scores at layer $k$ can be represented by $2(H+1)^{k-1}(2 \ell(n) + s)$ bits.
\end{lemma}

\begin{proof}
For an input sequence of length $n$, the initial activation values are $(\sigma,i,n)$, where $\sigma \in \Sigma \cup \{\$\}$ and $i \in [1..n]$. This can be represented by a string of $2 \ell(n) + s$ bits.  At each successive layer, the activation values are a tuple of $(H+1)$ values from the previous layer, which multiplies the number of bits required to represent them by $(H+1)$.
Also, the range of attention scores is bounded by the number of ordered pairs of activation values at the previous layer, so attention values can be represented by twice the number of bits to represent an activation value at the previous layer.
\end{proof}
It is worth observing that the bounds provided by \autoref{lemma:activation-attention-bit-bounds} do not hold in the case of AHAT. Attention scores may be the result of the average of an arbitrary subset of the possible inputs, which means that there are exponentially more possible activation values at each layer.

The following elementary facts about Boolean circuits will be useful.
\begin{lemma}
\label{lemma:b-inputs}
An arbitrary Boolean function $f$ of $n$ inputs and $m$ outputs can be computed by a depth $3$ circuit of size at most $m(n2^n + 2^n + n)$.
\end{lemma}
\begin{proof}
Express each output $z_i$ of $f$ as a disjunctive normal form (DNF) formula of at most $2^n$ terms, each with at most $n$ literals.
Convert each DNF formula to a circuit with one OR gate with inputs from an AND gate for each term, each of whose inputs is either an input to the function, or the result of applying a NOT gate to an input.
In each such circuit, the OR gate has at most $2^n$ input wires, each AND gate has at most $n$ input wires, and each of at most $n$ NOT gates has one input wire, for a total size bounded by $n2^n + 2^n + n$.  The final circuit consists of these $m$ separate circuits computing in parallel, and its size is at most $m$ times the size of each one.
The longest possible path to an output from an input is through a NOT, an AND, and the OR gate, for a depth of at most $3$.
\end{proof}

\begin{corollary}
\label{corollary:log-n-inputs}
If a Boolean function $f$ has at most $c \log n$ inputs and at most $d \log n$ outputs,
then it may be computed by a Boolean circuit of depth $3$
and size at most $(d\log n)(n^c(c\log n) + n^c + c\log n)$.
\end{corollary}

With the $O(\log n)$ bound on the number of bits to represent activation and attention values, \autoref{lemma:activation-attention-bit-bounds} yields circuits of constant depth and size polynomial in $n$ for the input, attention, activation, and output functions.
Additional circuitry is necessary to implement the comparison of attention scores and selection of the activation value to attend to for each position, layer, and head.

We construct the overall circuit $C_{sn}$ according to the layers of $T$, starting with the input function. Let the inputs to $T$ be $x_i$ for $i \in [1..n]$. The inputs to $C_{sn}$ are $x_{i,j}$ for $i \in [1..n]$ and $j \in [1..s]$, where $x_{i,j}$ are the bits of $h(x_i)$, representing the binary encoding of input symbol $x_i$.
At layer $0$ for position $i$, the value of $y_i^{(0)} = f(x_i,i,n) =(x_i,i,n)$ is achieved by having the input wires $x_{i,j}$ for $j \in [1..s]$ followed by a sequence of constants $0$ or $1$ representing $\bin(i,n)$ and $\bin(n,n)$ for a total of $2\ell(n) + s$ wires representing the value $(x_i,i,n)$.

Inducting on layers, we assume that for some $k \in [1..\layers]$ the circuit $C_{sn}$ has been constructed to contain the wires representing all the activation values $y_i^{(k-1)}$ for $i \in [1..n]$ at layer $k-1$. The portion of the circuit computing the representations of activation values at layer $k$ is described as follows. Fix a position $i \in [1..n]$ and a head $h \in [1..H]$. For each $j \in [1..n]$, there is a circuit $A_{i,j,k,h}$ that has as input the wires for the activation values $y_i^{(k-1)}$ and $y_j^{(k-1)}$ and as output, wires representing the nonnegative integer attention score $a_{i,j,k,h}$ in binary. Each of these circuits $A_{i,j,k,h}$ has $2(H+1)^{k-1}(2\ell(n) + s)$ inputs and outputs by \autoref{lemma:activation-attention-bit-bounds}, and therefore can be computed using depth $3$ and size polynomial in $n$, by \autoref{corollary:log-n-inputs}. All 
$Hn^2$ such circuits for layer $k$ operate in parallel, for overall depth $3$ and size polynomial in $n$.

We next describe the circuit that implements the pooling function $f^{\uha}$. 
For each pair $j, j' \in [1..n]$, there is a circuit $D_{i,j,j',k,h}$ 
whose inputs are the outputs of
$A_{i,j,k,h}$ and $A_{i,j',k,h}$ and whose output is a single wire $g_{i,j,j',k,h}$ with a value of $1$ if $a_{i,j,k,h} \ge a_{i,j',k,h}$ and $0$ otherwise.
Because of the bounds on the number of inputs and outputs, each of these circuits can have depth $3$ and size polynomial in $n$ by \autoref{corollary:log-n-inputs}.
These $n^2$ circuits all compute in parallel.\footnote{In fact, comparison of two $b$-bit integers can be done with a Boolean circuit of constant depth and size polynomial in $b$, but that is not necessary for the present purpose.} Then for each position $j$, whether $j$ maximizes $a_{i,j,k,h}$ can be computed by an AND gate whose inputs are $g_{i,j,j',k,h}$ for all $j' \in [1..n]$.  Let the output of this AND gate be denoted $m_{i,j,k,h}$.
Then $m_{i,j,k,h} = 1$ if and only if the position $j$ maximizes $a_{i,j,k,h}$.
This increases the depth by $1$.

For each $j$, an indicator $z_{i,j,k,h}$ is computed by an AND gate whose inputs are $m_{i,j,k,h}$ and $\text{NOT}(m_{i,j',k,h})$ for all $j' < j$.  Thus, $z_{i,j,k,h} = 1$ if and only if $j$ is the leftmost position that maximizes $a_{i,j,k,h}$.
This increases the depth by $2$.

Finally, these indicator values are used to combine the layer $k-1$ activation values in a selection circuit, yielding the representation of the activation value $b_{i,k,h} = y_j^{(k-1)}$ such that $z_{i,j,k,h} = 1$.
In general, such a selection circuit takes as input $t$ selector bits $z_1, \ldots, z_t$, where exactly one $z_j = 1$, and $t$ input values $w_1, \ldots, w_t$, where each $w_r$ consists of $S$ bits. It outputs $S$ bits representing the selected $w_j$ (for which $z_j = 1$).
Letting $w_{r,s}$ denote bit $s$ of $w_r$, the computation can be described as $v_{r,s} = w_{r,s} \wedge z_r$ for $r \in [1..t]$ and $s \in [1..S]$, which can be computed by one layer of $tS$ AND gates in parallel. 
Then the bits of the output are $u_s = \bigvee_{r = 1}^S v_{r,s}$ for $s \in [1..S]$, which can be computed by one layer of $S$ OR gates in parallel.
Thus, the selection circuit adds $2$ to the depth, and a polynomial in $n$ to the size.

Because each activation function for a GUHAT in informative normal form simply returns its arguments, no further computation is needed for the activation values.  The representation of the activation value $y_i^{(k)}$ is just the sequence of wires representing $y_i^{(k-1)}$ followed by those representing $b_{i,k,1}$, through $b_{i,k,H}$.

To produce the output of the circuit, we note that the representation of $y_n^{(L)}$ has $O(\log n)$ bits and the output of $g$ is a single bit, so $g$ can be implemented by a Boolean circuit of constant depth and size polynomial in $n$, by \autoref{corollary:log-n-inputs}. This concludes the proof of \autoref{theorem:guHAT0-in-AC0}.

\citet{FSS84} prove that the \parity, \equality, and \majority\ languages are not in AC$^0$, which immediately implies the following.
\begin{corollary}
\label{corollary:parity-and-majority-not-in-AC0}
GUHAT does not contain the languages \parity, \majority, or \equality.
\end{corollary}

To see that the \dyck-$k$ languages are also not in AC$^0$, we reduce from the \equality\ language.

\begin{corollary}
\label{Dyck-languages-not-in-AC0}
For all $k \ge 1$, the language \dyck-$k$ is not in AC$^0$, and is therefore not in GUHAT.
\end{corollary}

\begin{proof}
It suffices to prove this for $k = 1$.
Assume that there is a family $\{C_n\}$ of Boolean circuits in AC$^0$ that recognizes \dyck-$1$.  We may assume that the binary symbol encoding is $h([) = 0$ and $h(]) = 1$.
We show how to use this to construct a family $\{E_n\}$ of Boolean circuits in AC$^0$ that recognizes the \equality\ language, a contradiction.

$E_n$ is constructed from $C_{3n}$ as follows.  If the inputs to $E_n$ are $x_1, \ldots, x_n$, then $E_n$ consists of $C_{3n}$ with its first $n$ inputs set to the constant $0$, its middle $n$ inputs set to $x_1, \ldots, x_n$, and its last $n$ inputs set to the constant $1$.

Let $x$ be any element of $\{0,1\}^n$.
If the number of occurrences of $0$ is not equal to the number of occurrences of $1$ in $x$, then the input to $C_{3n}$ has unequal numbers of $[$ and $]$ symbols, which is not in \dyck-$1$ and $x$ is rejected.
If the number of occurrences of $0$ is equal to the number of occurrences of $1$ in $x$, then in any prefix of the input to $C_{3n}$, the number of occurrences of $]$ is less than or equal to the number of occurrences of $[$. At the end of the input, the number of occurrences of $[$ is equal to the number of occurrences of $]$, so the input to $C_{3n}$ is in \dyck-$1$ and $x$ is accepted.
Thus $\{E_n\}$ is a family of Boolean circuits in AC$^0$ that recognizes the language \equality, a contradiction.
\end{proof}

\section{Discussion and Conclusions}
\label{sec:conclusions}

We have defined formal language recognition by the encoder portion of a Transformer network using generalized unique hard attention (GUHAT), unique hard attention (UHAT), and averaging hard attention (AHAT), and shown that languages in UHAT and GUHAT are recognizable by constant depth, polynomial size families of circuits, that is, families of circuits in the complexity class AC$^0$.  This strengthens the negative result of \citet{Hahn20} that the languages \parity\ and \dyck-$k$ are not in GUHAT, and provides a simpler and more general proof. Combined with prior results of \citet{PMB2019} showing that the language \majority\ is in AHAT, or \citet{bhattamishra-etal-2020-ability} showing that the language \dyck-$1$ is in AHAT, this shows that AHAT contains languages that are not in GUHAT or UHAT.

Many intriguing open questions remain. What classical closure properties hold for the classes of languages GUHAT, UHAT, and AHAT?  Closure under complement just requires complementing the output function $g$, and closure under pairwise union and intersection should be straightforward using a parallel approach; but what about closure under homomorphism, inverse homomorphism, concatenation, or Kleene star? We briefly observe that GUHAT and UHAT cannot be closed under both Kleene star and concatentation lest they contain all regular languages, including \parity. 

Existing formal models and indeed practical implementations of Transformers vary in their representation of position information, whether as an absolute representation of position, a ratio (e.g., position $i$ in a sequence of length $n$ as  $i/n$), through angle information (e.g., position $i$ by the pair $(\cos \theta_i, \sin \theta_i)$ where $\theta_i = \pi i/2n$), or as an arbitrary learned embedding. In the UHAT and AHAT models, the choice of positional encoding can facilitate positional comparison (e.g., an angle-based encoding allows for equality testing via dot products) or make it uncomputable (e.g., if positional encodings enumerate Turing machines that halt on their own encodings). 
It remains to be understood what effect such differences in position representation have on the expressive power of a model.  

More generally, is it possible to prove that soft attention, which we have not addressed here, is strictly more powerful than even averaging hard attention?
\citet[Theorem B.3]{YPPN2021} present a construction for a soft attention Transformer that recognizes \dyck-$k$. This construction crucially employs specialized encodings of position and layer normalization, whose formal power remains to be understood.

Finally, given the success that Transformers have had as models of natural language, it is perhaps surprising that these models' expressive power seems to be best characterized (or at least bounded) in terms of circuit complexity. Mathematical explorations of natural language have most commonly employed the approach to language complexity afforded by the Chomsky hierarchy and its refinements, which is based on automata and formal grammars. The apparent incomparability of these approaches suggests that the exploration of different types of Transformer models might offer a new approach to the study of the formal properties of natural language.

\section*{Acknowledgements}

We thank the reviewers and the action editor for their work in reviewing this paper.

\bibliography{constant-depth}
\bibliographystyle{acl_natbib}

\end{document}